\documentclass[11pt]{article}
\usepackage{graphicx}
\usepackage{latexsym}
\usepackage{amssymb}
\usepackage{amsmath}
\usepackage{amsthm}
\usepackage{forloop}
\usepackage[paper=letterpaper,margin=1in]{geometry}
\usepackage[ruled,vlined,linesnumbered]{algorithm2e}
\usepackage{paralist}
\usepackage{nicefrac}
\usepackage{mleftright}

\newtheorem{theorem}{Theorem}
\newtheorem{lemma}{Lemma}[section]

\newtheorem{corollary}[lemma]{Corollary}

\newtheorem{observation}[lemma]{Observation}

\makeatletter
\newtheorem*{rep@theorem}{\rep@title}
\newcommand{\newreptheorem}[2]{%
\newenvironment{rep#1}[1]{%
 \def\rep@title{#2 \ref{##1}}%
 \begin{rep@theorem}}%
 {\end{rep@theorem}}}
\makeatother

\newreptheorem{lemma}{Lemma}

\newcommand{\defcal}[1]{\expandafter\newcommand\csname c#1\endcsname{{\mathcal{#1}}}}
\newcommand{\defbb}[1]{\expandafter\newcommand\csname b#1\endcsname{{\mathbb{#1}}}}
\newcounter{calBbCounter}
\forLoop{1}{26}{calBbCounter}{
    \edef\letter{\Alph{calBbCounter}}
    \expandafter\defcal\letter
		\expandafter\defbb\letter
}

\newcommand{\eps}{\varepsilon}
\newcommand{\ie}{{\it i.e.}}
\newcommand{\eg}{{\it e.g.}}
\newcommand{\nnR}{{\bR_{\geq 0}}}
\newcommand{\characteristic}{{\mathbf{1}}}
\newcommand{\RSet}{{\mathtt{R}}}
\newcommand{\partsol}[2][]{{\ifthenelse{\equal{#1}{}}{y^{(#2)}}{y^{(#2)}_{#1}}}}

\title{Guess Free Maximization of Submodular and Linear Sums}
\author{Moran Feldman\thanks{Department of Mathematics and Computer Science, The Open University of Israel. E-mail: moranfe@openu.ac.il}}

\begin{document}

\maketitle
\begin{abstract}
We consider the problem of maximizing the sum of a monotone submodular function and a linear function subject to a general solvable polytope constraint. Recently, Sviridenko et al.~\cite{SVW17} described an algorithm for this problem whose approximation guarantee is optimal in some intuitive and formal senses. Unfortunately, this algorithm involves a guessing step which makes it less clean and significantly affects its time complexity. In this work we describe a clean alternative algorithm that uses a novel weighting technique in order to avoid the problematic guessing step while keeping the same approximation guarantee as the algorithm of~\cite{SVW17}.

\medskip

\noindent \textbf{Keywords:} submodular maximization, continuous greedy, curvature
\end{abstract}
\pagenumbering{Alph}
\thispagestyle{empty}
\clearpage
\pagenumbering{arabic}

\pagenumbering{arabic}

\section{Introduction}

The last decade has seen a surge of work on submodular maximization problems. Arguably, the main factor that allowed this surge was the invention of the multilinear relaxation for submodular maximization problems as well as algorithms for (approximately) solving this relaxation~\cite{BF16,CCPV11,CVZ14,EN16,FNS11}. The invention of the multilinear relaxation was so influential because it allowed algorithms for submodular maximization to use the technique of first solving a relaxed version of the problem, and then rounding the fractional solution obtained. This technique is well-known, and it is often used in the design of algorithms for other kinds of problems; but, prior to the invention of the multilinear relaxation, it was not known how to apply it to submodular maximization problems.

An algorithm based on the above mentioned technique usually has two main components: a solver that (approximately) solves the relaxation and a rounding procedure. Historically, the first solver described for multilinear relaxations was the Continuous Greedy algorithm that solves such relaxations up to an approximation ratio of $1 - \nicefrac{1}{e}$ when the objective function is non-negative and monotone (in addition to being submodular)~\cite{CCPV11}.\footnote{A set function $f\colon 2^\cN \to \bR$ is \emph{monotone} if $f(S) \leq f(T)$ for every two sets $S \subseteq T \subseteq \cN$ and \emph{submodular} if $f(S \cup \{u\}) - f(S) \geq f(T \cup \{u\}) - f(T)$ for every two such sets and element $u \in \cN \setminus T$.} While the invention of continuous greedy was very significant, one can note that unlike standard solvers for more familiar relaxations such as LPs and SDPs, the approximation ratio of continuous greedy is quite far from $1$. Unfortunately, a hardness result due to~\cite{NW78} implies that its approximation ratio cannot be improved in general.

This situation motivates the question of how well can one approximate multilinear relaxations whose objective includes both monotone submodular and linear components. Specifically, it is interesting to know whether the approximation ratio that can be achieved in such cases improves gradually as the linear component of the objective becomes more prominent. Recently, Sviridenko et al.~\cite{SVW17} answered this question in the affirmative. More formally, they considered the following problem. Given a non-negative monotone submodular function $g\colon 2^\cN \to \nnR$, a linear function $\ell$ and a solvable polytope $P \subseteq [0, 1]^\cN$,\footnote{A polytope $P$ is solvable if one can optimize linear functions subject to it.} find a point $x \in P$ that approximately maximizes $G(x) + \ell(x)$, where $G$ is the multilinear extension of $G$ (see Section~\ref{sec:preliminaries} for a definition). Sviridenko et al.~\cite{SVW17} described a variant of continuous greedy that, given an instance of this problem, outputs a vector $x$ obeying the inequality $G(x) + \ell(x) \geq (1 - e^{-1}) \cdot g(OPT) + \ell(OPT)$ up to a small error term, where $OPT$ is an optimal integral solution for the problem.\footnote{Technically, Sviridenko et al.~\cite{SVW17} considered only the special case of the problem in which $P$ is a matroid polytope, and designed two algorithms for this case. However, one of these algorithms (the continuous greedy based one) trivially extends to arbitrary solvable polytopes.}

Intuitively, the result of Sviridenko et al.~\cite{SVW17} is tight since it approximates the submodular component of the objective up an approximation ratio of $1 - \nicefrac{1}{e}$ and the linear component up to a ratio of $1$. More formally, Sviridenko et al.~\cite{SVW17} showed that their result is tight since it yields optimal approximation ratios for two problems of interest: maximizing a non-negative monotone submodular function with a bounded curvature subject to a matroid constraint, and minimizing a non-negative non-increasing supermodular function with a bounded curvature subject to the same kind of constraint.

\subsection{Our Result}
Despite being optimal in terms of its approximation guarantee, in the senses described above, the algorithm of Sviridenko et al.~\cite{SVW17} suffers from a significant drawback. Namely, it is based on guessing the contribution of the linear component of the objective to the optimal solution, and this guessing step is quite problematic for the following reasons.
\begin{itemize}
	\item The guessing is done by enumerating $\Theta(n\eps^{-1} \log n)$ different possible values, and thus, increases the time complexity of the algorithm by this factor. Moreover, to guarantee that one of the enumerated values is a good enough guess, the set of values tried is constructed in a non-trivial way, which is then reflected in the complexity of the algorithm's analysis.
	\item The original continuous greedy algorithm repeatedly maximizes linear functions subject to the constraint polytope $P$. While this computational step is quite slow in general, for many cases of interest (for example, when $P$ is a matroid polytope) it can be implemented very efficiently. In contrast, the algorithm of~\cite{SVW17} maximizes linear functions subject to the intersection of $P$ with a polytope defined by the guessed value, which might be a very slow operation even when optimizing linear functions subject to $P$ itself is fast. Moreover, various techniques have been described to speed up continuous greedy when it is applied to a matroid polytope~\cite{BV14,BFS17}, and these techniques fail to apply to the algorithm of~\cite{SVW17} because it considers the intersection of $P$ with another polytope rather than $P$ itself.
\end{itemize}

In this work we present a clean alternative algorithm that has the same approximation guarantee as the algorithm of Sviridenko et al.~\cite{SVW17}, but avoids the guessing step and all the problems resulting from it. In a nut shell, our algorithm is a modification of continuous greedy in which the weight assigned to each component of the objective function varies over time. At the beginning of an execution of the algorithm, the linear component has much more weight than the monotone submodular component, and over time their weights become equal. Intuitively, this kind of weighting makes sense because the standard analysis of continuous greedy for submodular functions uses a lower bound on the gain of the algorithm in its later steps which decreases if the algorithm has already made a significant gain in earlier steps. Thus, any gain from the submodular component of the objective that is obtained early in the algorithm's execution is partially cancelled by the resulting decrease in the gain guaranteed in later steps of the execution. In contrast, gain obtained from the linear component of the objective in the same early steps of the execution does not suffer from such partial cancellation, and thus, should get more weight.

\subsection{Additional Related Work}

When the linear function is non-negative, its sum with the monotone submodular function is still monotone and submodular. Thus, in this case the work of Sviridenko et al.~\cite{SVW17} can be viewed as improving the gurantee of continuous greedy in a special case. More recently, Soma and Yoshida~\cite{SY17} used an algorithm based on a similar technique to improve over the guarantee of continuous greedy in the more general case in which the monotone submodular objective can be decomposed into a monotone submodular component and a significant $M^\natural$-concave component. In an earlier work, Feldman et al.~\cite{FNS11} took the complementing approach of using properties of the constraint polytope, rather than the objective function, to improve over the guarantee of continuous greedy. Specifically, they described a variant of continuous greedy, named Measured Continuous Greedy, which achieves an improved approximation ratio when the constraint is dense (in some sense).

As mentioned above, an algorithm that works by solving a relaxation and then rounding the solution has two main components: a relaxation solver and a rounding procedure. All the discussion up to this point was devoted to relaxation solvers because the current work is about such a solver and also because, unlike solvers, rounding procedures tend to be very problem specific. Nevertheless, there are a few more noticeable such procedures. A large portion of the work done so far on submodular maximization has been in the context of matroid constraints, for which there are two known rounding procedures that do not lose anything in the objective: Pipage Rounding~\cite{CCPV11} and Swap Rounding~\cite{CVZ10}. In another line of work, Chekuri et al.~\cite{CVZ14} designed a framework called ``contention resolution schemes'' which yields a rounding procedure for every constraint that can be presented as the intersection of few simple constraints. Later works extended the contention resolution schemes framework into online and stochastic settings~\cite{AW18,FSZ16,GN13}.
\section{Preliminaries} \label{sec:preliminaries}

In this section we describe the notation that we use and give a few relevant definitions. Using these definitions we then formally describe the guarantee of the algorithm we analyze.

Given a set $S$ and an element $u$, we use $S + u$ and $S - u$ as shorthands for the union $S \cup \{u\}$ and the expression $S \setminus \{u\}$, respectively. If we are also given a set function $f$, then the marginal contribution of $u$ to $S$ with respect to $f$ is denoted by $f(u \mid S) \triangleq f(S + u) - f(S)$. Notice that using this notation we get that a function $f\colon 2^\cN \to \bR$ is submodular if and only if for every two sets $S \subseteq T \subseteq \cN$ and element $u \in \cN \setminus T$ it holds that $f(u \mid S) \geq f(u \mid T)$. Occasionally, we are also interested in the marginal contribution of a set $T$ to a set $S$ with respect to $f$, which we denote by $f(T \mid S) \triangleq f(S \cup T) - f(S)$.

The \emph{multilinear} extension of a set function $f\colon 2^\cN \to \bR$ is a function $F: [0, 1]^\cN \to \bR$ whose value for a vector $x \in [0, 1]^\cN$ is defined as
\[
	F(x)
	=
	\bE[f(\RSet(x))]
	=
	\sum_{S \subseteq \cN} \mleft(\prod_{u \in S} x_u \mright) \cdot \mleft(\prod_{u \in \cN \setminus S} (1 - x_u) \mright) \cdot f(S)
	\enspace,
\]
where $\RSet(x)$ is a random set containing every element $u \in \cN$ with probability $x_u$, independently. One can observe that $F$ is an extension of $f$ in the sense that for every set $S \subseteq \cN$, if we denote by $\characteristic_S$ the characteristic vector of $S$, then it holds that $f(S) = F(\characteristic_S)$. Additionally, observe that the rightmost hand side of $F$'s definition implies that $F$ is indeed a multilinear function, as suggested by its name. The multilinearity of $F$ implies that for every vector $y \in [0, 1]^\cN$ and element $u \in \cN$ the partial derivative of $F$ with respect to $y$ is given by
\[
	\left.\frac{\partial F(x)}{\partial x_u}\right|_{x = y}
	=
	F(y \vee \characteristic_{\{u\}}) - F(y \wedge \characteristic_{\cN - u})
	=
	\bE[f(u \mid \RSet(y) - u)]
	\enspace,
\]
where the vector operations $\vee$ and $\wedge$ represent coordinate-wise maximum and minimum, respectively. %In many places we also need coordinate-wise multiplication of vectors, which we denote by $\odot$.

A linear function is defined by a vector $\ell \in \bR^\cN$. We abuse notation and identify the vector $\ell$ with the linear function it defines. Accordingly, we denote the value of the function for a vector $x \in [0, 1]^\cN$ as $\ell(x) \triangleq \ell \cdot x$. In a further abuse of notation, given a set $S \subseteq \cN$, we use $\ell(S)$ as a shorthand for $\ell(\characteristic_S) = \sum_{u \in S} \ell_u$.

An instance of the problem we consider in this work consists of a non-negative monotone submodular function $g\colon 2^\cN \to \nnR$, a linear function $\ell\colon 2^\cN \to \bR$ and a solavable polytope $P \subseteq [0, 1]^\cN$. We make the standard assumption that the submodular function $g$, whose description might be exponential in terms of the size $n$ of $\cN$, is accessible to the algorithm through a value oracle that given a set $S \subseteq \cN$ returns $f(S)$. The objective  of the problem is to find a vector $x \in P$ maximizing $G(x) + \ell(x)$, where $G$ is the multilinear extension of $g$. The result that we prove for this problem is given by the next theorem. Let $OPT$ be the set corresponding to an optimal integral solution for the problem, \ie, $OPT = \arg\max_{S \subseteq 2^\cN, \characteristic_S \in P} \{g(S) + \ell(S)\}$, and let $m = \max_{u \in \cN} \{g(u \mid \varnothing)\}$.

\begin{theorem} \label{thm:main_result} 
There exists a polynomial time algorithm for the above problem that given a value $\eps \in (0, 1)$ outputs a vector $x \in P$ such that with high probability $G(x) + \ell(x) \geq (1 - e^{-1}) \cdot g(OPT) + \ell(OPT) - O(\eps) \cdot m$.
\end{theorem}

Theorem~\ref{thm:main_result} is very similar to the corresponding result of Sviridenko et al.~\cite{SVW17}. However, there are two differences between the two. First, Sviridenko et al.~\cite{SVW17} considered only matroid polytopes, for which there are known lossless rounding methods~\cite{CCPV11,CVZ10}, and thus, their result is stated in terms of sets rather than vectors. Second, the error term of~\cite{SVW17} depends also on $\max_{u \in \cN} |\ell_u|$, which is unnecessary for the analysis of our cleaner algorithm.
\section{Algorithm}

In this section we give a non-formal proof of Theorem~\ref{thm:main_result}. This proof demonstates our new ideas, but uses some non-formal simplifications such as
allowing a direct oracle access to the multilinear extension
$G$ of $g$ and giving the algorithm in the form
of a continuous time algorithm (which cannot be implemented
on a discrete computer).  There are known techniques for getting rid of these simplifications (see, \eg,~\cite{CCPV11}), and for completeness, Appendix~\ref{sec:DiscreteAlgorithm} includes a formal proof of Theorem~\ref{thm:main_result} based on these techniques.

The algorithm we use for the non-formal proof of Theorem~\ref{thm:main_result} is given as Algorithm~\ref{alg:distorted_continuous_greedy}. Like the original continuous greedy algorithm of~\cite{CCPV11}, this algorithm grows a solution $y(t)$ over time. The solution starts as $\characteristic_\varnothing$ at time $t = 0$, and the output of the algorithm is the solution at time $t = 1$. Our algorithm differs, however, from the original continuous greedy algorithm in the method used to determine the direction in which the solution is grown at every given time point. Specifically, our algorithm defines a weight vector $w(t)$ for every time $t \in [0, 1)$ based on the derivatives of the multilinear extension $G$. It then looks for a vector $z(t)$ in $P$ maximizing a \emph{weighted} combination of $w(t)$ with the linear function $\ell$, and this vector $z(t)$ determines the direction in which the solution $y(t)$ is grown.

\begin{algorithm}[h]
\caption{\textsf{Distorted Continuous Greedy}($g, \ell, P$)} \label{alg:distorted_continuous_greedy}
\DontPrintSemicolon
Let $y(0) \gets \characteristic_\varnothing$.\\
\ForEach{time $t \in [0, 1)$}
{
	For each $u \in \cN$, let $w_u(t) \gets \left.\frac{\partial G(x)}{\partial x_u}\right|_{x = y(t)}$.\\
	Let $z(t)$ be the vector in $P$ maximizing $z(t) \cdot (e^{t-1} \cdot w(t) + \ell)$.\\
	Increase $y(t)$ at a rate of $\frac{dy(t)}{dt} = z(t)$.
}
\Return{$y(1)$}.
\end{algorithm}

We begin the analysis of Algorithm~\ref{alg:distorted_continuous_greedy} by observing that its output is a vector in $P$.
\begin{observation} \label{obs:feasibility}
$y(1) \in P$.
\end{observation}
\begin{proof}
By definition, $z(t)$ is a vector in $P$ for every time $t \in [0, 1)$. Hence, $y(1) = \int_0^1 z(t)dt$ is a convex combination of vectors in $P$, and thus, belongs to $P$ by the convexity of $P$.
\end{proof}

Let us consider now the function $\Phi(t) \triangleq e^{t - 1} \cdot G(y(t)) + \ell(y(t))$. This function is a central component in our analysis of the approximation ratio of Algorithm~\ref{alg:distorted_continuous_greedy}. The following technical lemma gives an expression for the derivative of this important function.
\begin{lemma} \label{lem:derivative}
\[
	\frac{d\Phi(t)}{dt}
	=
	e^{t - 1} \cdot G(y(t)) + z(t) \cdot (e^{t - 1} \cdot w(t) + \ell)
	\enspace.
\]
\end{lemma}
\begin{proof}
By the chain rule,
\begin{align*}
	\frac{d\Phi(t)}{dt}
	={}&
	e^{t - 1} \cdot G(y(t)) + e^{t - 1} \cdot \frac{dG(y(t))}{dt} + \frac{d\ell(y(t))}{dt}\\
	={} &
	e^{t - 1} \cdot G(y(t)) + e^{t - 1} \cdot \sum_{u \in \cN} \frac{dy_u(t)}{dt} \cdot \left.\frac{\partial G(x)}{\partial x_u}\right|_{x = y(t)} + \sum_{u \in \cN} \frac{dy_u(t)}{dt} \cdot \left.\frac{\partial \ell(x)}{\partial x_u}\right|_{x = y(t)} \\
	={} &
	e^{t - 1} \cdot G(y(t)) + e^{t - 1} \cdot \sum_{u \in \cN} z_u(t) \cdot w_u(t) + \sum_{u \in \cN} z_u(t) \cdot \ell_u\\
	={} &
	e^{t - 1} \cdot G(y(t)) + z(t) \cdot (e^{t - 1} \cdot w(t) + \ell)
	\enspace.
	\qedhere
\end{align*}
\end{proof}

The next lemma lower bounds the expression given by the last lemma for the derivative of $\Phi(t)$.
\begin{lemma} \label{lem:derivative_explained}
For every $t \in [0, 1)$,
\[
	e^{t - 1} \cdot G(y(t)) + z(t) \cdot (e^{t - 1} \cdot w(t) + \ell)
	\geq
	e^{t - 1} \cdot g(OPT) + \ell(OPT)
	\enspace.
\]
\end{lemma}
\begin{proof}
Recall that $z(t)$ is chosen by Algorithm~\ref{alg:distorted_continuous_greedy} as the vector in $P$ maximizing $z(t) \cdot (e^{t - 1} \cdot w(t) + \ell)$. Since $\characteristic_{OPT} \in P$, we get
\[
	z(t) \cdot (e^{t - 1} \cdot w(t) + \ell)
	\geq
	\characteristic_{OPT} \cdot (e^{t - 1} \cdot w(t) + \ell)
	=
	e^{t - 1} \cdot \sum_{u \in OPT} \mspace{-9mu} w_u(t) + \ell(OPT)
	\enspace.
\]

Note now that, by the multilinearity of $G$,
\begin{align*}
	\sum_{u \in OPT} \mspace{-9mu} w_u(t)
	={} &
	\sum_{u \in OPT} \mspace{-9mu} \left.\frac{\partial G(x)}{\partial x_u}\right|_{x = y(t)}
	=
	\sum_{u \in OPT} \mspace{-9mu} [G(y(t) \vee \characteristic_{\{u\}}) - G(y(t) \wedge \characteristic_{\cN - u})]\\ \nonumber
	\geq{} &
	\sum_{u \in OPT} \mspace{-9mu} [G(y(t) \vee \characteristic_{\{u\}}) - G(y(t))]
	\geq
	G(y(t) \vee \characteristic_{OPT}) - G(y(t))
	\geq
	g(OPT) - G(y(t))
	\enspace,
\end{align*}
where the first and last inequalities follow from the monotonicity of $g$, and the remaining inequality follows from its submodularity.

Combining the two above inequalities yields
\[
	z(t) \cdot (e^{t - 1} \cdot w(t) + \ell)
	\geq
	e^{t - 1} \cdot [g(OPT) - G(y(t))] + \ell(OPT)
	\enspace,
\]
and the lemma now follows by adding $e^{t-1} \cdot G(y(t))$ to both sides of this inequality.
\end{proof}

We are now ready to prove Theorem~\ref{thm:main_result}.

\begin{proof}[Proof of Theorem~\ref{thm:main_result}]
Observation~\ref{obs:feasibility} shows that $y(1) \in P$. Thus, to prove the theorem it only remains to prove $G(y(1)) + \ell(y(1)) \geq (1 - e^{-1}) \cdot g(OPT) + \ell(OPT)$. 

Lemmata~\ref{lem:derivative} and~\ref{lem:derivative_explained} prove together that
\[
	\frac{d\Phi(t)}{dt}
	=
	e^{t - 1} \cdot G(y(t)) + z(t) \cdot (e^{t - 1} \cdot w(t) + \ell)
	\geq
	e^{t - 1} \cdot g(OPT) + \ell(OPT)
	\enspace.
\]
Integrating both sides of this inequality from $t = 0$ to $t = 1$, we get
\[
	\Phi(1) - \Phi(0)
	\geq
	(1 - e^{-1}) \cdot g(OPT) + \ell(OPT)
	\enspace,
\]
and the theorem now follows by noticing that
\[
	\Phi(1)
	=
	G(y(1)) + \ell(y(1))
	\qquad
	\text{and}
	\qquad
	\Phi(0)
	=
	e^{-1} \cdot G(y(0)) + \ell(y(0))
	\geq
	0
\]
(the last inequality holds since $g$ is non-negative and $\ell(y(0)) = \ell(\characteristic_\varnothing) = 0$).
\end{proof}

\bibliographystyle{plain}
\bibliography{Curvature}

\appendix
\section{Formal Proof of Theorem~\ref{thm:main_result}} \label{sec:DiscreteAlgorithm}

In this section we give a formal proof of Theorem~\ref{thm:main_result}. The algorithm that we use for this proof is given as Algorithm~\ref{alg:formal}. Notice that this algorithm considers only discrete times that are integer multiples of a value $\delta$ chosen in a way guaranteeing two things. First, that $\delta \leq \min\{\nicefrac{1}{2}, \eps n^{-2}\}$, and second, that $1$ is an integer multiple of $\delta$.

\begin{algorithm}[h]
\caption{\textsf{Distorted Continuous Greedy -- Formal}($g, \ell, P, \eps$)} \label{alg:formal}
\DontPrintSemicolon
Let $y(0) \gets \characteristic_\varnothing$, $t \gets 0$ and $\delta \gets \lceil 2 + n^2/\eps \rceil^{-1}$.\\
\While{$t < 1$}
{
	For each $u \in \cN$, let $w_u(t)$ be an estimate for $\bE[g(u \mid \RSet(y(t)) - u)]$ obtained by averaging the value of the expression within this expectation for $r = \lceil -2n^2\eps^{-2} \ln (\delta/n^2) \rceil$ independent samples of $\RSet(y(t))$.\\
	Let $z(t)$ be the vector in $P$ maximizing $z(t) \cdot [(1 + \delta)^{(t-1)/\delta} \cdot w(t) + \ell]$.\\
	Let $y(t + \delta) \gets y(t) + \delta \cdot z(t)$.\\
	Update $t \gets t + \delta$.
}
\Return{$y(1)$}.
\end{algorithm}

Let $T$ be the set of times considered by Algorithm~\ref{alg:formal}, \ie, $T = \{i\delta \mid i \in \bZ, 0 \leq i < \delta^{-1}\}$. The following observation, which corresponds to Observation~\ref{obs:feasibility}, shows that the output of Algorithm~\ref{alg:formal} is feasible.
\begin{observation} \label{obs:feasibility_formal}
$y(1) \in P$.
\end{observation}
\begin{proof}
By definition, $z(t)$ is a vector in $P$ for every time $t \in T$. Observe also that $|T| = \delta^{-1}$. Hence,
$
	y(1) = \sum_{t \in T} \delta \cdot z(t)
$
is a convex combination of vectors in $P$, and thus, belongs to $P$ by the convexity of $P$.
\end{proof}

Next, we need to lower bound the probability that any of the estimates made by Algorithm~\ref{alg:formal} has a significant error. This is done by Lemma~\ref{lem:error_probability}, whose proof is based on the following known lemma.
\begin{lemma}[The symmetric version of Theorem A.1.16 in~\cite{AS00}] \label{lem:concentration}
Let $X_i$, $1 \leq i \leq k$, be mutually independent with all $\bE[X_i] = 0$ and all $|X_i| \leq 1$. Set $S = X_1 + \dotsb + X_k$. Then, $\Pr[|S| > a] \leq 2e^{-a^2/2k}$.
\end{lemma}

Let $\cE$ be the event that
\[
	|w_u(t) - \bE[g(u \mid \RSet(y(t)) - u)]|
	\leq
	\frac{\eps m}{n}
\]
for every every element $u \in \cN$ and time $t \in T$. 
\begin{lemma} \label{lem:error_probability}
$\Pr[\cE] \geq 1 - 2n^{-1}$, and hence, $\cE$ is a high probability event.
\end{lemma}
\begin{proof}
Consider an arbitrary element $u \in \cN$ and time $t \in T$, and let us denote by $R_i$ the $i$-th independent sample of $\RSet(y(t))$ used for calculating $w_u(t)$. We now define for every $1 \leq i \leq r$
\[
	X_i
	=
	\frac{[g(u \mid R_i - u)] - \bE[g(u \mid \RSet(y(t)) - u)]}{m}
	\enspace.
\]

Clearly $\bE[X_i] = 0$ due to the linearity of the expectation. Additionally, note that $X_i \in [-1, 1]$ because the monotonicity of $g$ guarantees that $g(u \mid R_i - u)$ and $\bE[g(u \mid \RSet(y(t)) - u)]$ are both non-negative, and the submodularity of $g$ guarantees that these expressions are upper bounded by $f(u \mid \varnothing) \leq m$. Thus, by Lemma~\ref{lem:concentration},
\begin{align*}
	\Pr\mleft[|w_u(t) - \bE[g(u \mid \RSet(y(t)) - u)]|
	>
	\frac{\eps m}{n}\mright]
	={} &
	\Pr\mleft[\frac{m}{r} \cdot \mleft|\sum_{i=1}^r X_i\mright|
	>
	\frac{\eps m}{n}\mright]
	=
	\Pr\mleft[\mleft|\sum_{i=1}^r X_i\mright|
	>
	\frac{r\eps}{n}\mright]\\
	\leq{} &
	2e^{-(r\eps n^{-1})^2/2r}
	=
	2e^{-r\eps^2/(2n^2)}
	\leq
	2e^{\ln (\delta / n^2)}
	=
	\frac{2\delta}{n^2}
	\enspace.
\end{align*}

Using the union bound, we now get that the probability that there is any pair of element $u \in \cN$ and time $t \in T$ for which
\[
	|w_u(t) - \bE[g(u \mid \RSet(y(t)) - u)]|	> \frac{\eps m}{n}
\]
is at most
\[
	|\cN| \cdot |T| \cdot \frac{2\delta}{n^2}
	=
	\frac{2}{n}
	\enspace.
	\qedhere
\]
\end{proof}

Let us define now $\Phi(t) \triangleq (1 + \delta)^{(t-1)/\delta} \cdot G(y(t)) + \ell(y(t))$. Lemma~\ref{lem:derivative_formal} bounds the rate in which this expression increases as a function of $t$ (and thus, can be viewed as a counterpart of Lemma~\ref{lem:derivative}). The following technical lemma is used in the proof of Lemma~\ref{lem:derivative_formal}. Since similar lemmata have been proved in other places (see, for example,~\cite{F17,FNS11}), we defer the proof of this lemma to Appendix~\ref{app:step_improvement}.
\begin{lemma} \label{lem:step_improvement}
Given two vectors $y, y' \in [0, 1]^\cN$ such that $0 \leq y'_u - y_u \leq \delta \leq 1$ and a non-negative monotone submodular function $f\colon 2^\cN \to \nnR$ whose multilinear extension is $F$,
\[
	F(y') - F(y)
	\geq
	\sum_{u \in \cN} (y'_u - y_u) \cdot \left.\frac{\partial F(x)}{\partial x_u}\right|_{x = y} - n^2\delta^2 \cdot \max_{u \in \cN} f(u \mid \varnothing)
	\enspace.
\]
\end{lemma}

\begin{lemma} \label{lem:derivative_formal}
If the event $\cE$ happens, then, for every time $t \in T$,
\[
	\frac{\Phi(t + \delta) - \Phi(t)}{\delta}
	\geq
	(1 + \delta)^{(t - 1)/\delta} \cdot G(y(t)) + z(t) \cdot [(1 + \delta)^{(t - 1)/\delta} \cdot w(t) + \ell] - 2\eps m
	\enspace.
\]
\end{lemma}
\begin{proof}
Since $y(t + \delta) - y(t) = \delta z(t)$, and every coordinate of $z(t)$ is between $0$ and $1$, we get by Lemma~\ref{lem:step_improvement} that
\begin{align} \label{eq:G_diff}
	G(y(t + \delta)) - G(y(t))
	\geq{} &
	\sum_{u \in \cN} (y(t + \delta) - y(t)) \cdot \left.\frac{\partial G(x)}{\partial x_u}\right|_{x = y(t)} - n^2\delta^2 \cdot \max_{u \in \cN} g(u \mid \varnothing)\\ \nonumber
	\geq{} &
	\sum_{u \in \cN} \delta z_u(t) \cdot \bE[g(u \mid \RSet(y(t)) - u)] - \eps\delta m\\ \nonumber
	\geq{} &
	\sum_{u \in \cN} \delta z_u(t) \cdot [w_u(t) - \eps m/n] - \eps\delta m
	\geq
	\delta z(t) \cdot w(t) - 2\eps\delta m
	\enspace,
\end{align}
where the second inequality hold since $\delta \leq \eps n^{-2}$ by definition, and the third inequality holds since we assume that the event $\cE$ happened.

Using the linearity of $\ell$ and the definition of $\Phi$, we now get
\begin{align*}
	&
	\frac{\Phi(t + \delta) - \Phi(t)}{\delta}\\
	={} &
	\frac{[(1 + \delta)^{(t + \delta - 1)/\delta} \cdot G(y(t + \delta)) - (1 + \delta)^{(t - 1)/\delta} \cdot G(y(t))] + [\ell(y(t + \delta)) - \ell(y(t))]}{\delta}\\
	={} &
	\frac{(1 + \delta)^{(t - 1)/\delta} \cdot [G(y(t + \delta)) - G(y(t))] }{\delta} + (1 + \delta)^{(t - 1)/\delta} \cdot G(y(t + \delta)) + \ell \cdot z(t)
	\enspace.
\end{align*}
Plugging Inequality~\eqref{eq:G_diff} and the inequality $G(y(t + \delta)) \geq G(y(t))$ (which holds due to monotonicity) into the last equality, we get
\[
	\frac{\Phi(t + \delta) - \Phi(t)}{\delta}
	\geq
	(1 + \delta)^{(t - 1)/\delta} \cdot G(y(t)) + z(t) \cdot [(1 + \delta)^{(t - 1)/\delta} \cdot w(t) + \ell] - 2\eps m \cdot (1 + \delta)^{(t - 1)/\delta}
	\enspace.
\]
The lemma now follows from the last inequality by observing that $(1 + \delta)^{(t - 1)/\delta }\leq 1$ since $(t - 1)/\delta \leq 0$.
\end{proof}

The last lemma gives a lower bound on the increase in $\Phi(t)$ as a function of $t$. Unfortunately, this lower bound depends on a lot of entities (such as $z(t)$ and $w(t)$), and thus, it is difficult to use it. The following lemma allows us to simplify the lower bound.
\begin{lemma}
If the event $\cE$ happens, then
$z(t) \cdot [(1 + \delta)^{(t - 1)/\delta} \cdot w(t) + \ell] \geq (1 + \delta)^{(t - 1)/\delta} \cdot [g(OPT) - G(y(t))] + \ell(OPT) - \eps m$.
\end{lemma}
\begin{proof}
Recall that $z(t)$ is the vector in $P$ maximizing $z(t) \cdot [(1 + \delta)^{(t-1)/\delta} \cdot w(t) + \ell]$. Since $\characteristic_{OPT} \in P$, we get
\begin{align*}
	z(t) \cdot [(1 + \delta)^{(t-1)/\delta} \cdot w(t) + \ell]
	\geq{} &
	\characteristic_{OPT} \cdot [(1 + \delta)^{(t-1)/\delta} \cdot w(t) + \ell]\\
	={} &
	(1 + \delta)^{(t-1)/\delta} \cdot \sum_{u \in OPT} \mspace{-9mu} w(u) + \ell(OPT)\\
	\geq{} &
	(1 + \delta)^{(t-1)/\delta} \cdot \sum_{u \in OPT} \mspace{-9mu} \{\bE[g(u \mid \RSet(y(t)) - u)] - \eps m /n\} + \ell(OPT)\\
	\geq{} &
	(1 + \delta)^{(t-1)/\delta} \cdot \sum_{u \in OPT} \mspace{-9mu} \bE[g(u \mid \RSet(y(t)) - u)] + \ell(OPT) - \eps m
	\enspace,
\end{align*}
where the second inequality holds since we assume that the event $\cE$ happened. Observe now that the submodularity and monotonicity of $f$ yield
\begin{align*}
	\sum_{u \in OPT} \mspace{-9mu} \bE[g(u \mid \RSet(y(t)) - u)]
	\geq{} &
	\sum_{u \in OPT} \mspace{-9mu} \bE[g(u \mid \RSet(y(t)))]
	\geq
	\bE[g(OPT \mid \RSet(y(t)))]\\
	={} &
	\bE[g(OPT \cup \RSet(y(t))) - g(\RSet(y(t)))]
	\geq
	g(OPT) - G(y(t))
	\enspace.
\end{align*}
The lemma now follows by combining the two above inequalities.
\end{proof}

Combining the last two lemmata, we immediately get the following corollary, which is the promised simplified lower bound on the increase in $\Phi(t)$ as a function of $t$.
\begin{corollary} \label{cor:combined_bound_formal}
\[
	\frac{\Phi(t + \delta) - \Phi(t)}{\delta}
	\geq
	(1 + \delta)^{(t - 1)/\delta} \cdot g(OPT) + \ell(OPT) - 3\eps m
	\enspace.
\]
\end{corollary}

Using the last corollary, we can now get a lower bound on the value of $G(y(1)) + \ell(y(t))$ conditioned on the event $\cE$.
\begin{lemma} \label{lem:bound_given_event}
If the event $\cE$ happens, then $G(y(1)) + \ell(y(1)) \geq (1 - e^{-1}) \cdot g(OPT) + \ell(OPT) - 4\eps m$.
\end{lemma}
\begin{proof}
Observe that
\begin{align} \label{eq:basic_bound}
	G(y(1)) + \ell(y(1))
	={} &
	\Phi(1)
	=
	\Phi(0) + \sum_{t \in T} [\Phi(t + \delta) - \Phi(t)]\\ \nonumber
	\geq{} &
	\Phi(0) + \sum_{t \in T} \mleft[\delta (1 + \delta)^{(t - 1)/\delta} \cdot g(OPT) + \delta \cdot \ell(OPT) - 3\eps\delta m\mright]\\ \nonumber
	\geq{} &
	\sum_{t \in T} \delta(1 + \delta)^{(t - 1)/\delta} \cdot g(OPT) + \ell(OPT) - 3\eps m
	\enspace,
\end{align}
where the first inequality holds due to Corollary~\ref{cor:combined_bound_formal} and the second inequality holds since $|T| = \delta^{-1}$ and $\Phi(0) = (1 + \delta)^{-1/\delta} \cdot g(\varnothing) + \ell(\varnothing) = (1 + \delta)^{-1/\delta} \cdot g(\varnothing) \geq 0$ because $\ell$ is linear and $g$ is non-negative.

We now need to lower bound the sum on the rightmost hand side of the last inequality. Notice that this sum can be presented as the sum of a geometrical series as follows.
\begin{align*}
	\sum_{t \in T} \delta(1 + \delta)^{(t - 1)/\delta}
	={} &
	\sum_{i = 0}^{\delta^{-1} - 1} \delta(1 + \delta)^{(i\delta - 1) / \delta}
	=
	\delta(1 + \delta)^{-\delta^{-1}} \cdot \sum_{i = 0}^{\delta^{-1} - 1} (1 + \delta)^i\\
	={} &
	\delta(1 + \delta)^{-\delta^{-1}} \cdot \frac{1 - (1 + \delta)^{\delta^{-1}}}{1 - (1 + \delta)}
	=
	1 - (1 + \delta)^{-\delta^{-1}}\\
	\geq{} &
	1 - e^{-1}(1 - \delta)^{-1}
	\geq
	1 - e^{-1}(1 + 2\delta)
	\geq
	1 - e^{-1} - \eps / n
	\enspace,
\end{align*}
where the first inequality holds since it is known that $(1 + 1/a)^a \geq e(1 - 1/a)$ for every $a \geq 1$ (and in particular for $a = \delta ^{-1}$), the second inequality holds since $(1 - a)^{-1} \leq 1 + 2a$ for every $a \leq \nicefrac{1}{2}$, and the last inequality holds since $\delta \leq \eps/n^2$.
%
%Since $(1 + \delta)^{(t - 1)/\delta}$ is an increasing function of $t$, we also get
%\begin{align*}
	%\sum_{t \in T} \delta(1 + \delta)^{(t - 1)/\delta}
	%={} &
	%(1 + \delta)^{-1} \cdot \sum_{t \in T} \delta(1 + \delta)^{(t + \delta - 1)/\delta}
	%\geq
	%(1 + \delta)^{-1} \cdot \int_0^1 (1 + \delta)^{(x - 1)/\delta}\\
	%={} &
	%\frac{\delta}{(1 + \delta)\ln(1 + \delta)} \cdot \mleft[(1 + \delta)^{(x - 1)/\delta}\mright]_0^1
	%=
	%\frac{\delta}{(1 + \delta)\ln(1 + \delta)} \cdot \mleft[1 - (1 + \delta)^{- 1/\delta}\mright]\\
	%\geq{} &
	%\frac{\delta}{\ln(1 + \delta)} \cdot \mleft[(1 + \delta)^{-1} - (1 + \delta)^{-1 - 1/\delta}\mright]
	%\geq
	%\frac{\delta}{\ln(1 + \delta)} \cdot \mleft[(1 + \delta)^{-1} - e^{-1}\mright]
	%\enspace.
%\end{align*}
%Using the known inequalities $x \geq \ln(1 + x)$ and $(1 + x)^{-1} \geq 1 - x$, the last inequality becomes
%\[
	%\sum_{t \in T} \delta(1 + \delta)^{(t - 1)/\delta}
	%\geq
	%1 - \delta - e^{-1}
	%\geq
	%1 - e^{-1} - \eps/n
	%\enspace,
%\]
%where the last inequality holds since $\delta \leq \eps/n^2$ by definition.
The lemma now follows by plugging the last inequality into Inequality~\eqref{eq:basic_bound} and observing that, by the submoduarity of $g$,
\[
	g(OPT)
	\leq
	g(\varnothing) + \sum_{u \in OPT} g(u \mid \varnothing)
	\leq
	g(\varnothing) + mn
	\enspace.
	\qedhere
\]
\end{proof}

We are now ready to prove Theorem~\ref{thm:main_result}.
\begin{proof}[Proof of Theorem~\ref{thm:main_result}]
Observation~\ref{obs:feasibility_formal} shows that $y(1) \in P$. Additionally, Lemmata~\ref{lem:error_probability} and~\ref{lem:bound_given_event} show together that with high probability
\[
	G(y(1)) + \ell(y(1))
	\geq
	(1 - e^{-1}) \cdot g(OPT) + \ell(OPT) - 4\eps m
	=
	(1 - e^{-1}) \cdot g(OPT) + \ell(OPT) - O(\eps) \cdot m
	\enspace.
	\qedhere
\]
\end{proof}
\section{Proof of Lemma~\ref{lem:step_improvement}} \label{app:step_improvement}

In this section we prove Lemma~\ref{lem:step_improvement}. Let us begin by recalling the lemma itself.

\begin{replemma}{lem:step_improvement}
Given two vectors $y, y' \in [0, 1]^\cN$ such that $0 \leq y'_u - y_u \leq \delta \leq 1$ and a non-negative monotone submodular function $f\colon 2^\cN \to \nnR$ whose multilinear extension is $F$,
\[
	F(y') - F(y)
	\geq
	\sum_{u \in \cN} (y'_u - y_u) \cdot \left.\frac{\partial F(x)}{\partial x_u}\right|_{x = y} - n^2\delta^2 \cdot \max_{u \in \cN} f(u \mid \varnothing)
	\enspace.
\]
\end{replemma}

Let us denote the elements of $\cN$ by $u_1, u_2, \dotsc, u_n$ in an arbitrary order. We define $\partsol{i}$ for every integer $0 \leq i \leq n$ as the vector in $[0, 1]^\cN$ that agrees with $y'$ on the coordinates $1$ to $i$ and with $y$ or the remaining coordinates. Note that this definition implies, in particular, $\partsol{0} = y$ and $\partsol{n} = y'$. The next lemma bounds the amount by which the partial derivative $\frac{\partial F(x)}{\partial x_u}$ can differ between the points $x = y$ and $x = y'$.
\begin{lemma} \label{lem:derivative_change}
For every integer $0 \leq i \leq n$ and element $u \in \cN$,
\[
	\left.\frac{\partial F(x)}{\partial x_u}\right|_{x = \partsol{i}}
	\geq
	\left.\frac{\partial F(x)}{\partial x_u}\right|_{x = y} - n\delta \cdot f(u \mid \varnothing)
	\enspace.
\]
\end{lemma}
\begin{proof}
For the sake of the proof, we assume that $\RSet(\partsol{i})$ is formed from $\RSet(y)$ using the following process. Every element of $\cN \setminus \RSet(y)$ is added to a set $D$ with probability of $1 - (1 - \partsol[u]{i}) / (1 - y_u)$. Then, $\RSet(\partsol{i})$ is chosen as $\RSet(y) \cup D$. Observe that every element $u \in \cN$ gets into $D$ with probability $\partsol[u]{i} - y_u \leq \delta$, independently, and thus, $\RSet(y) \cup D$ indeed has the distribution that $\RSet(\partsol{i})$ should have.

Using the above definitions, we get
\begin{align*}
	\left.\frac{\partial F(x)}{\partial x_u}\right|_{x = \partsol{i}}
	={} &
	\bE[f(u \mid \RSet(\partsol{i}) - u)]
	=
	\bE[f(u \mid \RSet(y) \cup D - u)]\\
	\geq{} &
	\Pr[D = \varnothing] \cdot \bE[f(u \mid \RSet(y) - u) \mid D = \varnothing]
	%\geq{} &
	%\Pr[D = \varnothing] \cdot \bE[f(\RSet(x) + u) - f(\RSet(x) - u) \mid D = \varnothing]
	%- \Pr[D \neq \varnothing] \cdot n \cdot \max_{u \in \cN} f(u)
	\enspace,
\end{align*}
where the inequality follows from the law of total expectation and the monotonicity of $f$. Additionally, by the submodularity of $f$ we also get
\[
	f(u \mid \RSet(y) - u)
	\leq
	f(u \mid \varnothing)
	\enspace.
\]
Combining this inequality with the previous one yields
\begin{align*}
	\left.\frac{\partial F(x)}{\partial x_u}\right|_{x = \partsol{i}} &{}+ \Pr[D \neq \varnothing] \cdot f(u \mid \varnothing)\\
	\geq{} &
	\Pr[D = \varnothing] \cdot \bE[f(u \mid \RSet(y) - u) \mid D = \varnothing]
	+
	\Pr[D \neq \varnothing] \cdot \bE[f(u \mid \RSet(y) - u) \mid D \neq \varnothing]\\
	={} &
	\bE[f(u \mid \RSet(y) - u)]
	=
	\left.\frac{\partial F(x)}{\partial x_u}\right|_{x = y}
	\enspace.
\end{align*}

One can verify that the last inequality will imply the lemma if we have an upper bound of $n\delta$ on $\Pr[D \neq \varnothing]$. Thus, all we are left to do is to prove this upper bound. Since elements belong to $D$ with probability at most $\delta$ and independently,
\[
	\Pr[D \neq \varnothing]
	=
	1 - \prod_{u \in \cN} \Pr[u \not \in D]
	\leq
	1 - \prod_{u \in \cN} (1 - \delta)
	=
	1 - (1 - \delta)^n
	\leq
	n\delta
	\enspace.
	\qedhere
\]
\end{proof}

We are now ready to prove Lemma~\ref{lem:step_improvement}.
\begin{proof}[Proof of Lemma~\ref{lem:step_improvement}]
Observe that for every integer $1 \leq i \leq n$ the vectors $\partsol{i-1}$ and $\partsol{i}$ differ only in coordinate $i$ (in which they differ by $y'_u - y_u$). Recalling that $\partsol{0} = y$, $\partsol{n} = y'$ and $F$ is multilinear, this observation yields
\begin{align*}
	F(y') - F(y)
	={} &
	\sum_{i = 1}^n (y'_{u_i} - y_{u_i}) \cdot \left. \frac{\partial F(x)}{\partial x_{u_i}}\right|_{x = \partsol{i-1}}\\
	\geq{} &
	\sum_{u \in \cN} (y'_u - y_u) \cdot \mleft[\left. \frac{\partial F(x)}{\partial x_{u}}\right|_{x = y} - n\delta \cdot f(u \mid \varnothing)\mright]\\
	\geq{} &
	\sum_{u \in \cN} (y'_u - y_u) \cdot \left. \frac{\partial F(x)}{\partial x_{u}}\right|_{x = y} - n^2\delta^2 \cdot \max_{u \in \cN} f(u \mid \varnothing)
	\enspace,
\end{align*}
where the first inequality follows from Lemma~\ref{lem:derivative_change} and the second inequality holds by the monotonicity of $f$ and the fact that $y'_u - y_u \leq \delta$ for every $u \in \cN$.
\end{proof}

\end{document}